\documentclass[a4paper,11pt]{article}

\usepackage[english]{babel}
\usepackage{authblk,fullpage}
\usepackage{amsthm,amsmath,graphicx,url,pdfsync,enumerate,booktabs,xcolor,amssymb}        
\allowdisplaybreaks

\newtheorem{theorem}{Theorem}[section]
\newtheorem{definition}{Definition}[section]
\newtheorem{proposition}{Proposition}[section]
\newtheorem{remark}{Remark}[section]
\newtheorem{lemma}{Lemma}[section]

\title{\bf Optimizing edge weights in the inverse eigenvector centrality problem}

\author[1]{Mauro Passacantando\footnote{Corresponding author.}}

\author[2]{Fabio Raciti}

\affil[1]{\small University of Milano-Bicocca, Department of Business and Law, Via Bicocca degli Arcimboldi 8, 20126 Milan, Italy, \texttt{mauro.passacantando@unimib.it}}

\affil[2]{Department of Mathematics and Computer Science, University of Catania, Viale A. Doria 6, 95125 Catania, Italy, \texttt{fabio.raciti@unict.it}}

\date{}

\begin{document}

\maketitle

\noindent\textbf{Abstract:} 
In this paper we study the inverse eigenvector centrality problem on directed graphs: given a prescribed node centrality profile, we seek edge weights that realize it. Since this inverse problem generally admits infinitely many solutions, we explicitly characterize the feasible set of admissible weights and introduce six optimization problems defined over this set, each corresponding to a different weight-selection strategy. These formulations provide representative solutions of the inverse problem and enable a systematic comparison of how different strategies influence the structure of the resulting weighted networks.
We illustrate our framework using several real-world social network datasets, showing that different strategies produce  different weighted graph structures while preserving the prescribed centrality. The results highlight the flexibility of the proposed approach and its potential applications in network reconstruction, and network design or network manipulation.

\



\section{Introduction}
\label{sec:intro}

The concept of centrality plays a fundamental role in network analysis, as it provides a quantitative measure of the relative importance of nodes in complex systems. 
In social networks, centrality is used to identify influential individuals, opinion leaders, and key actors in information diffusion; in engineering systems, it helps detect critical components affecting robustness and resilience; while, in biological networks it is commonly employed to uncover essential genes or proteins, often associated with vital functions or disease mechanisms.
Depending on the application context, different notions of centrality have been proposed. Some measures are purely topological in nature, such as degree, closeness, and betweenness, while others -- such as eigenvector centrality, Katz–Bonacich centrality \cite{Bon87}, and PageRank \cite{BriPag98} -- are based on recursive definitions or arise from underlying dynamical processes on the network. The latter class assigns importance to nodes through their interactions with other important nodes, thereby capturing global patterns of influence and connectivity. 
The seminal work of Freeman~\cite{Fre79} introduced a unified framework for centrality measures in social networks, while a comprehensive overview of these concepts and their applications is provided in the Newman’s book~\cite{New10}. 
Given the pervasive and growing role of networks across disciplines, the study of centrality measures remains an active area of research, also in connection with related fields such as game theory (see, e.g., \cite{PasRac26}).

We focus here on the eigenvector  centrality, which  assigns importance to nodes according to the importance of their neighbors, making it especially suitable for modeling influence, prestige, and propagation phenomena, and motivating growing interest in inverse and design problems based on such measures. 
In the seminal paper~\cite{Lat_etal12}, the authors showed that, for strongly connected directed networks, any prescribed centrality profile can be realized by assigning (in infinite possible ways) suitable weights to the edges. They also investigated the minimum controlling set problem, which consists in identifying the smallest set of nodes capable of inducing edge weights compatible with the prescribed centrality. 
The same research line has been pursued in~\cite{Cipolla25} with focus on Katz and PageRank centrality measures.

We thus  consider six optimization problems over the feasible set of weights which can be viewed as encoding different mechanisms for resolving the inherent non-uniqueness of the inverse centrality problem. When a regulator is present, each formulation may be interpreted as a distinct intervention strategy, whereby edge weights are selected so as to realize the prescribed centrality while optimizing a specific system-level criterion. In the absence of a central authority, the same formulations can alternatively be understood as modeling different endogenous mechanisms through which network interactions are shaped, reflecting implicit preferences or constraints that guide the emergence of particular weighted structures. This viewpoint provides a unifying motivation for the proposed optimization framework, allowing us to compare how different design principles or behavioral assumptions lead to qualitatively different realizations of the same target centrality.

The paper is organized as follows. 
In Section~\ref{sec:recall}, we introduce the notation and review some basic graph-theoretic notions, together with the definition of eigenvector centrality and  a brief overview of the Perron–Frobenius theorem. 
In Section~\ref{sec:inverse}, we present the inverse eigenvector centrality problem. 
In Section~\ref{sec:optim}, we describe our optimization problems and derive {\em a priori} bounds for their solutions. 
In Section~\ref{numexp}, we illustrate our models on a small graph and three real-world social network datasets. The concluding section summarizes our main findings and outlines directions for future research.


\section{Notation and preliminaries}
\label{sec:recall}

This section recalls the basic material needed for the rest of the paper. We first briefly review essential graph-theoretic concepts and matrix properties and  then introduce eigenvector centrality and state the Perron-Frobenius theorem for different classes of matrices.


\subsection{Recall on Graph Theory and Matrix Classes}
\label{sec:recal}

We consider here directed  graphs $G=(V,E)$, where $V=\{1,\dots,n\}$ is the set of nodes (or vertices) and $E \subset V \times V$, with $m=|E|$, is the set of arcs formed by ordered pairs of nodes $(i,j)$; $i$ is called the tail of the arc while $j$ is the head.
Moreover, we consider simple graphs, that is, there are no multiple arcs connecting the same pair of nodes, nor loops.
Two nodes $i$ and $j$ are said to be adjacent if they are connected by the arc $(i,j)$. 
The degree of node $i$ in a graph, {\em deg(i)}, is the number of arcs connected to it. 
In a directed graph, we distinguish between in-degree (number of incoming arcs) and out-degree (number of outgoing arcs).
The set of tails of the incoming arcs in a node $j$ is denoted by $BS(j) = \{ i \in N : (i,j) \in E \}$.
Information about the adjacency of nodes can be stored in the adjacency matrix $A$ whose elements $a_{ij}$ are equal to $1$ if $(i,j)$ is an edge, $0$ otherwise. 
Given a node $i$, the nodes connected to $i$ with an arc are called the {\em neighbors} of $i$.

An ordered sequence of nodes \((i_0, i_1, \ldots, i_r)\), with \(r \ge 1\), is called a \emph{directed walk} from \(i_0\) to \(i_r\) 
if \((i_h, i_{h+1}) \in {E}\) for \(h = 0, 1, \ldots, r-1\).
The node \(i_0\) is called the \emph{origin} of the directed walk, while \(i_r\) is its 
\emph{destination}.
The {\em length} of a walk is the number of its arcs. 
Let us remark that it is allowed to visit a node or go through an arc more than once.
A directed graph is {\em strongly connected} if for each pair of vertices $i$ and $j$ there is a directed walk from $i$ to $j$ and a directed walk from $j$ to $i$.
This relation is an equivalence relation on $V$. 
The equivalence classes
induced by it are called the \emph{strongly connected components} of $G$.
A strongly connected component is therefore a maximal subset
$C\subseteq V$ such that every pair of nodes in $C$ is mutually reachable.
When one strongly connected component contains a significantly larger number
of nodes than the others, it is commonly referred to as the
\emph{giant strongly connected component}.

A weighted  graph $(V,E,w)$ is a graph in which each arc  is assigned a numerical value called a weight.
These weights typically represent quantities such as distance, cost, time, capacity, or strength of connection. 
In the adjacency matrix of a weighted graph, denoted by $A_w$, the weights replace the ones. 
We use the notation $\mathbf{1}$ to denote the all-ones vector, with the dimension understood from the context.

The  connections between any two nodes in the graph are described by means of the powers of the adjacency matrix $A$, as specified in the following theorem.

\begin{theorem}\label{potenzematrice}
Let $A \in \mathbb{R}^{n\times n}$ be the adjacency matrix of a graph. 
Then $(A^{k})_{ij}$ represents the number of walks of length k from node i to node j.
\end{theorem}

\begin{definition}
A matrix $A\in \mathbb{R}^{n\times n}$ is called nonnegative (positive) if all its elements are nonnegative (positive).
A matrix $P \in \mathbb{R}^{n\times n}$ is called a permutation matrix if in each row and column exactly one element is equal to 1 and all others are equal to $0$.
\end{definition}
Multiplication by a permutation matrix produces a permutation of the rows or columns of the multiplied matrix.

\begin{definition}
A matrix $A\in \mathbb{R}^{n\times n}$ is reducible if there exists a permutation matrix $P \in \mathbb{R}^{n\times n}$ such that
\[
P^{\top}AP=\left[ 
\begin{array}{cc}
B & C \\ 
0_{n-r,r} & D
\end{array} 
\right], 
\qquad 1\leq r \leq n-1. 
\]
A matrix $A\in \mathbb{R}^{n\times n}$ is irreducible if it is not reducible.
\end{definition}

\begin{proposition}
(see~\cite[Theorem 6.2.24 (d)]{Horn})
\\
The adjacency matrix of a directed graph is irreducible if and only if the graph is strongly connected.
\label{irriducibile}
\end{proposition}

For the subsequent development the following definition is fundamental.
\begin{definition}[Spectral radius]
\ \\
Let $A\in\mathbb{C}^{n\times n}$ and denote by $\sigma(A)$ the spectrum of $A$ (i.e., the set of its eigenvalues).
The \emph{spectral radius} of $A$ is defined as
\[
\rho(A) := \max\{ |\lambda| : \lambda \in \sigma(A) \}.
\]
\end{definition}

\subsection{Eigenvector Centrality and the Perron-Frobenius Theorem}

The idea of assigning importance scores via an eigenvector of a connection matrix has a long history.
The earliest known formulation appears in a paper by {Edmund Landau} (1895) on the ranking of players in chess tournaments \cite{Lan95},
where player strengths are obtained as the dominant eigenvector of the matrix of game outcomes.
The concept was later reinterpreted in social network analysis by Phillip Bonacich  in 1972 \cite{Bon72},
who introduced eigenvector centrality as a measure of status or influence in a network.
Bonacich’s formulation established the recursive principle that a node’s importance increases with the importance of its neighbors
that is,  a node is considered important if it is connected to other important nodes.
Let $G = (V,E)$ be a graph (or digraph) with adjacency matrix $A = (a_{ij}) \in \mathbb{R}^{n \times n}$.
The eigenvector centrality assigns to each vertex $i \in V$ a score $c_i$ proportional to the sum of the centralities of its neighbors:
\[
c_i = \frac{1}{\rho} \sum_{j=1}^{n} a_{ji} \, c_j,
\]
or, in compact form,
$A^{\top} c = \rho c$. This intuitive concept is formalized by the following definition.

\begin{definition}[Eigenvector Centrality]\label{ecdef}
\ \\
Let $G=(V,E)$ be a directed  graph with adjacency matrix
$A=(a_{ij})\in\mathbb{R}^{n\times n}$. Assume that $A$ is irreducible (equivalently,
that $G$ is strongly connected), and let $\rho>0$ denote the spectral radius of $A$.
The \emph{eigenvector centrality} of $G$ is the vector $c\in\mathbb{R}^n_{>0}$
defined as the unique (up to scaling) positive solution of
\begin{equation}
A^{\top}c=\rho c .
\end{equation}
The vector $c$ is normalized so that
$
\sum_{i=1}^n c_i = 1 .
$
The component $c_i$ is called the eigenvector centrality of vertex $i\in V$.
Existence, positivity, and uniqueness (up to normalization) of $c$ follow from
the Perron-Frobenius theorem. This definition is well posed because the Perron-Frobenius theorem also implies that there are no other positive eigenvectors apart from $c$.
\end{definition}

In the following, we state several versions of the Perron-Frobenius theorem. The version relevant to the eigenvector centrality problem is that for nonnegative irreducible matrices; the remaining ones are included for reference.
The presentation of this section is based on \cite[Chapter 8]{Horn}. The interest reader can also refer to \cite{BerPle94} and \cite{Min88}.

\begin{theorem}[Perron-Frobenius Theorem for positive matrices]
	\ \\
Let $A \in \mathbb{R}^{n\times n}$ be positive. Then, the following properties hold true:
\begin{enumerate}
\item $\rho(A)>0$;
\item  $\rho(A)$ is an algebraically simple eigenvalue of  $A$;
\item $\exists ! \ x \in \mathbb{R}^{n}$ such that $Ax= \rho(A)x$ and $ ||x||_{1}=1$; such vector is positive and is called Perron vector;
\item $\exists ! \ y \in \mathbb{R}^{n}$ such that $y^{\top} A= \rho(A)y^{\top}$ and $ x^{\top}y=1$; such vector is positive and is called left Perron vector;
\item $\rho (A)>|\lambda| $  for each eigenvalue  $\lambda$ of $A$ such that $\lambda \neq  \rho (A)$;
\item $\left( \dfrac{A}{\rho (A)}\right)^{m} \rightarrow  xy^{\top}$ for $m \to + \infty$.
\end{enumerate}
\end{theorem}

\begin{theorem}[Perron-Frobenius Theorem for nonnegative matrices]
	\ \\
Let $A \in~\mathbb{R}^{n\times n}$ be nonnegative. Then, the following properties hold true:
\begin{enumerate}
\item $\rho(A)>0$;
\item $\exists \, x \in \mathbb{R}^{n}\backslash{\left\lbrace 0 \right\rbrace}$ such that $Ax= \rho(A)x$ and $x \geq 0$.
\end{enumerate}
\end{theorem}

\begin{theorem}[Perron-Frobenius Theorem for nonnegative and irreducible matrices]
	\ \\
Let $A \in \mathbb{R}^{n\times n}$  be nonnegative and irreducible. 
Then, the following properties hold true:
\begin{enumerate}
\item $\rho(A)>0$;
\item  $\rho(A)$ is an algebraically simple eigenvalue of  $A$;
\item $\exists ! \ x \in \mathbb{R}^{n}_+$ such that $Ax= \rho(A)x$ and $ ||x||_{1}=1$; such an eigenvector is positive and is called  Perron vector;
\item $\exists ! \ y \in \mathbb{R}^{n}$ such that $y^{\top}A= \rho(A)y^{\top}$ and $ x^{\top}y=1$; such an eigenvector is positive and is called a left Perron vector.
\end{enumerate}
\end{theorem}

\begin{remark}
For directed graphs, eigenvector centrality is meaningful only for strongly connected networks, since otherwise the dominant eigenvector need not be unique or strictly positive. Strong connectivity guarantees that all nodes participate in the recursive definition of importance and ensures a well-defined centrality measure. When the graph is not strongly connected, one may alternatively restrict attention to its giant strongly connected component, when such a component exists.
\end{remark}


\section{Inverse Eigenvector Centrality Problem}
\label{sec:inverse}

We introduce the inverse eigenvector centrality problem with a simple example.
Let us consider the following directed and strongly connected graph. 
On the left is the graphical representation, on the right its adjacency matrix $A$.

\begin{figure}[htb]
\begin{minipage}{0.4\textwidth}
\begin{center}
	\includegraphics[width=0.6\textwidth]{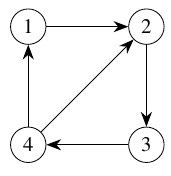}
\end{center}
\end{minipage}%
\begin{minipage}{0.6\textwidth}
{\large    \[
    A = \begin{pmatrix}
    0 & 1 & 0 & 0 \\
    0 & 0 & 1 & 0 \\
    0 & 0 & 0 & 1 \\
    1 & 1 & 0 & 0
    \end{pmatrix} 
    \]
}
\end{minipage}
\caption{Directed graph and corresponding adjacency matrix.}
\label{fig:example_graph}
\end{figure}

To find the eigenvector centrality, we solve $A^\top c = \rho c$.
Numerical calculation provides the eigenvalue $\rho_0 = 1.2207$ and the centrality eigenvector
\[ 
c_0 = (0.1808, \ 0.3290, \ 0.2695, \ 0.2207).
 \]
Consider now the weighted adjacency matrix $A_w$, assign arbitrarily $c$, (with strictly positive components) and consider the equation $A_w^\top c = \rho c$, where $\rho$ is a positive parameter, and the arc weights $w$ are the unknowns. Let us observe that in order to keep the topology of the network unchanged, we only consider strictly positive weights. The system  reads:
$$
\begin{cases}
c_4 w_{41} = \rho c_1 \\
c_1 w_{12} + c_4 w_{42} = \rho c_2 \\
c_2 w_{23} = \rho c_3 \\
c_3 w_{34} = \rho c_4
\end{cases}
$$
where the entries of $w$ are ordered by scanning the rows of $A$ and listing the indices corresponding to its nonzero entries. The above system has thus the form $Bw=\rho c$, where $B$ is given by
\[
B = \begin{pmatrix}
0 & 0 & 0 & c_4 & 0 \\
c_1 & 0 & 0 & 0 & c_4 \\
0 & c_2 & 0 & 0 & 0 \\
0 & 0 & c_3 & 0 & 0
\end{pmatrix}.
\]
The columns of $B$ correspond to the arcs, and each column contains only one non-zero element, located at the position of the arc's destination (head) and with a value equal to the centrality of the arc's source (tail).
 $B$ has rank 4 (since $c_i > 0$) and there are 4 columns with non-zero elements positioned in different rows, therefore the system admits $\infty^1$ solutions regardless of how the parameter $\rho > 0$ is fixed.
The solution  of the system $Bw=\rho c$ is
\[
\left( w_{12}, \ \rho \frac{c_3}{c_2}, \ \rho \frac{c_4}{c_3}, \ \rho \frac{c_1}{c_4}, \ \frac{\rho c_4 - c_1 w_{12}}{c_4} \right),
\]
which gives positive vectors  for every $w_{12}$ such that $0 < w_{12} < \rho c_4/c_1$. 

\

We now generalize the construction of the above example. 
Given a weighted graph with weighted adjacency matrix $A_w$, a fixed  centrality vector $c$ and a parameter $\rho >0$, the equation
\begin{equation}\label{eigw}
A^{\top}_w c=\rho c
\end{equation}
can be re-written as
\begin{equation}
B w= \rho c,
\end{equation}
where
\[
B_{ik} = \begin{cases} c_j & \text{if  arc } \alpha_k = (j, i) \\ 0 & \text{otherwise} \end{cases} \quad \implies \quad (Bw)_i = \sum_{j \in BS(i)} c_j w_{ji}
\]
or also, for column vectors:
\[
B_{\cdot, k} = c_j e_i, \quad \text{if } \alpha_k = (j, i)
\]
where $e_i$ is the $i-th$ canonical vector of $\mathbb{R}^n$.
Let us verify that this rule is correct by calculating the $i$-th component of the vector $Bw$:
\[
(Bw)_i 
= 
\sum_{k=1}^{m} B_{ik} w_k 
= 
\sum_{k: \alpha_k = (j, i)} c_j w_k 
= 
\sum_{j \in BS(i)} c_j w_{ji}
=
(A_w^\top c)_i .
\]
We now prove the existence of positive solutions of the system $Bw = \rho c$,
Let us consider the $j$-th row of the system $B w = \rho c$,
which we can rewrite as
\begin{align}\label{e:Bw=rhoc}
 \sum_{i \in BS(j)} c_i w_{ij} = \rho c_j, 
 \qquad \forall j \in N .
 \end{align}
Since every node has at least one incoming arc, we have $|BS(j)| \ge 1$.
For each node $j$, choose an incoming arc $(i^*, j)$, fix $\varepsilon > 0$ and set
\[
w_{ij} = \varepsilon, \qquad \forall\ i \in BS(j) \setminus \{ i^*\}.
\]
 Now define the weight of the arc $(i^*, j)$ by imposing
\[
c_{i^*} w_{i^* j} + \sum_{i \in BS(j) \setminus \{i^*\}} c_i \varepsilon = \rho c_j
\]
from which we derive
\[
w_{i^* j} 
= 
\frac{1}{c_{i^*}} \left[ \rho c_j - \varepsilon \sum_{i \in BS(j) \setminus \{i^*\}} c_i \right].
\]
Therefore, having fixed $\rho > 0$ and the vector $c > 0$, we have that $w_{i^* j} > 0$ if
\[
0 < \varepsilon < \frac{\rho c_j}{\displaystyle \sum_{i \in BS(j) \setminus \{i^*\}} c_i}.
\]
This condition must hold for all nodes $j$; therefore, to have $w_{i^* j} > 0$ for any $j$, one must have:
\[
0 < \varepsilon < \min_{j \in N} \frac{\rho c_j}{\displaystyle \sum_{i \in BS(j)} c_i}.
\]
In summary, given a prescribed centrality vector $c>0$ and a parameter $\rho>0$,
the inverse eigenvector centrality problem consists in finding arc weights
$w>0$ such that
\[
Bw=\rho c .
\]
As shown above, this system admits infinitely many positive solutions. In the
next section, we introduce six optimization formulations that select particular
solutions according to different design criteria.


\section{Optimization Problems}
\label{sec:optim}

In the following problems we look for solutions $w$ such that $w \geq \varepsilon \mathbf{1}$ for sufficiently small $\varepsilon >0$, in order to preserve the topology of the original network, while formulating the constraints using non-strict inequalities. Implementation details will be discussed in the numerical section.

\subsection{Problem Formulations}

\paragraph{(P1)\quad Minimum-$\ell_1$-norm formulation.}
\ \\
This strategy promotes sparsity in the deviation from the initial weights. It seeks solutions where as many weights as possible remain close or  equal to 1, concentrating the changes on a small number of arcs to achieve the target centrality.
\begin{equation*} \label{P1}
\begin{cases}
    \min\limits_{w} \ & \|w- \mathbf{1}\|_1 \\
    \text{s.t.} & Bw = \rho c \\
    				& w \geq \varepsilon \mathbf{1}
\end{cases}
\end{equation*} 

\paragraph{(P2) \quad Minimum-$\ell_2$-norm formulation.}
\ \\
This approach may represent a ``fairness''  strategy. By penalizing the square of the deviations, it discourages large changes in individual arc weights. As a consequence, we expect a network where the necessary adjustments are spread more evenly across all available arcs.
\begin{equation*} \label{P2}
\begin{cases}
    \min\limits_{w} \ & \|w -\mathbf{1}\|_2^2 \\
    \text{s.t.} & Bw = \rho c \\
    				& w \ge \varepsilon \mathbf{1}
\end{cases}
\end{equation*} 

\paragraph{(P3) \quad Minimum-$\ell_{\infty}$-norm formulation.}
\ \\
This is a conservative or ``min-max'' strategy. It focuses on minimizing the worst-case deviation. It is ideal in scenarios where there is a strict capacity limit or a cost threshold on how much any single arc weight can be modified.
\begin{equation*} \label{P3}
\begin{cases}
    \min\limits_{w} \ & \|w -{\mathbf{1}}\|_{\infty} \\
    \text{s.t.} & Bw = \rho c \\
    				& w \ge \varepsilon \mathbf{1}
\end{cases}
\end{equation*} 

\paragraph{(P4)\quad  Minimum-linear objective.}
\ \\
From a physical perspective, this formulation minimizes the total energy or resource consumption of the system, where $\beta_{ij}$ represents a unit cost of each link. In the context of social networks, it represents the minimization of the total social effort; if $\beta_{ij}$ is the cost of maintaining a relationship, (P4) identifies the most efficient distribution of interaction intensities $w_{ij}$ to satisfy the required influence levels with minimum collective investment.
\begin{equation*} \label{P4}
\begin{cases}
    \min\limits_{w} \ & \displaystyle \sum_{(i,j) \in E} \beta_{ij} w_{ij} \\
    \text{s.t.} & Bw = \rho c \\
    				& w \geq \varepsilon \mathbf{1}
\end{cases}
\end{equation*} 

\paragraph{(P5)\quad  Minimization of the number of nodes having at least one outgoing arc with weight $\neq 1$.}
\ \\
This model focuses on node-level efficiency. It is suitable for scenarios where changing an outgoing link's weight requires an administrative or technical cost for the source node, thus aiming to involve the minimum number of ``active'' nodes in the update.
In addition to the variables $w_{ij}$ for any $(i,j) \in E$, we introduce the auxiliary variables $x_{i} \in \{0,1\}$ for any $i \in N$, where $x_{i}=1$ if at least one outgoing arc from node $i$ has weight different from 1, and 0 otherwise. 
We consider the problem
\begin{equation*}\label{P5_model}
\left\{
\begin{array}{lll}
    \min\limits_{w,x}\ & \displaystyle \sum_{i \in N} x_{i} & 
    \\[4mm]
    \text{s.t.}\quad & B w = \rho c &
    \\[2mm]
    & 1 + (\varepsilon - 1) x_{i} \le w_{ij} \le 1 + M x_{i} & \qquad \forall (i,j) \in E
\end{array}
\right.
\end{equation*} 
where the parameter $M$ is set to
\begin{equation}\label{Emme}
	M = 
	\rho \dfrac{\max_{i \in N} c_i}{\min_{i \in N} c_i} - 1.
\end{equation} 
The big-M constraints guarantee that $w_{ij}=1$ for any outgoing arc from $i$ when $x_i=0$, while $\varepsilon \leq w_{ij} \leq 1+M$ holds when $x_i=1$. 
Notice that the chosen value for $M$ is not too restrictive as the following inequality holds for any $(i,j) \in E$ and $w$ such that $Bw=\rho c$:
\begin{align*}
	& w_{ij} 
	= 
	\dfrac{c_i \, w_{ij}}{c_i}
	\leq 
	\dfrac{\sum_{k \in BS(j)} c_k w_{kj}}{\min_{k \in N} c_k}
	=
	\dfrac{\rho c_j}{\min_{k \in N} c_k}
	\leq
	\dfrac{\rho\max_{k \in N} c_k}{\min_{k \in N} c_k}
	=
	1 + M.
\end{align*}

\paragraph{(P6) \quad Minimization of the number of arcs with weights $\neq 1$.}
\ \\
This MILP formulation targets the structural sparsity of the intervention. 
It is used when the goal is to modify the weights of as few links as possible, keeping the rest of the network's original interaction strengths intact. 
Similarly to (P5), it employs binary variables and big-M constraints ensuring the model finds the smallest set of structural changes necessary to satisfy the system's requirements.
We use variables $w_{ij}$ for any $(i,j) \in E$ and $x_{ij} \in \{0,1\}$ for any $(i,j) \in E$, where $x_{ij}=1$ if arc $(i,j)$ has weight $w_{ij} \neq 1$ and 0 otherwise. We consider the problem
\begin{equation*} \label{P6_model}
\left\{
\begin{array}{lll}
\min\limits_{w,x}\ & \displaystyle \sum_{(i,j) \in E} x_{ij} &
\\
\text{s.t.}\quad & B w = \rho c &
\\
& 1 + (\varepsilon - 1) x_{ij} \leq w_{ij} \leq 1 + M x_{ij} & \qquad \forall (i,j) \in E
\end{array}
\right.
\end{equation*} 
where $M$ is the same as in \eqref{Emme}. 

\begin{remark}
Problems (P1) and (P3) can be recast as linear programming problems through standard techniques. 
(P2) is a convex quadratic program, while (P5) and (P6) are mixed-integer linear programs. 
Moreover, we note that all the above problems except (P5) can be decomposed into $n$ problems, one for each node $j \in N$.
\end{remark}

\subsection{Solution bounds}
We denote with $W$ the set of positive solutions of the inverse problem, i.e., 
$$
	W = 
	\left\{ 
	w \in \mathbb{R}^m:\ 
	B \, w = \rho \, c, \quad w \geq \varepsilon\mathbf{1} 
	\right\}.
$$
We already observed, in the description of the inverse problem, that $W \neq \emptyset$ provided that $BS(j) \neq \emptyset$ for any $j \in N$.

\begin{lemma}\label{l:bound}
For any $w \in W$ the following inequalities hold:
\begin{enumerate}[a)]
\item $w_{ij} \leq \rho c_j/c_i$ for any $(i,j) \in E$;

\item $$
\dfrac{\rho \,c_j}{\max\limits_{i \in BS(j)} c_i}
\leq
\sum_{i \in BS(j)} w_{ij}
\leq
\dfrac{\rho \,c_j}{\min\limits_{i \in BS(j)} c_i}
\qquad \forall\ j \in N.
$$
\end{enumerate}
\end{lemma}

\begin{proof} \
\begin{enumerate}[a)]
 
\item For any arc $(i,j) \in E$ we have
$c_i w_{ij} \leq \sum_{k \in BS(j)} c_k w_{kj} = \rho \, c_j$.

\item It follows from~\eqref{e:Bw=rhoc} that
\begin{align*}
	\rho \, c_j 
	= 
	\sum_{i \in BS(j)} c_i w_{ij}
	\leq 
	\left[ \max\limits_{i \in BS(j)} c_i \right] \sum_{i \in BS(j)} w_{ij},
	\\
	\rho \, c_j 
	= 
	\sum_{i \in BS(j)} c_i w_{ij}
	\geq 
	\left[ \min\limits_{i \in BS(j)} c_i \right] \sum_{i \in BS(j)} w_{ij}.
\end{align*}

\end{enumerate}
\end{proof}

\begin{theorem}
For any $w \in W$ the following inequalities hold:
\begin{enumerate}[a)]

\item $$
\sum_{j \in N} 
\dfrac{\left| \rho \, c_j - \sum\limits_{i \in BS(j)} c_i \right|}{\max\limits_{i \in BS(j)} c_i}
\leq
\| w - \mathbf{1} \|_{1} 
\leq 
\sum_{(i,j) \in E} \max \left\{ 1 , \ \dfrac{\rho c_j}{c_i} - 1 \right\},
$$

\item $$
\sum_{j \in N} 
\dfrac{\left( \rho \, c_j - \sum\limits_{i \in BS(j)} c_i \right)^2}{ \sum\limits_{i \in BS(j)} c_i^2}
\leq
\| w - \mathbf{1} \|^2_{2} 
\leq 
m + 
\sum_{j \in N} \left[ \dfrac{\rho \,c_j}{\min\limits_{i \in BS(j)} c_i} \right]^2 
- 2 \sum_{j \in N} \dfrac{\rho \,c_j}{\max\limits_{i \in BS(j)} c_i},
$$

\item $$
\max_{j \in N} 
\dfrac{\left| \rho \, c_j - \sum\limits_{i \in BS(j)} c_i \right|}{\sum\limits_{i \in BS(j)} c_i}
\leq
\| w - \mathbf{1} \|_{\infty} 
\leq 
\max \left\{ 1 , \ \max_{(i,j) \in E} \left[ \dfrac{\rho c_j}{c_i} - 1 \right]\right\},
$$

\item $$
\rho \sum_{j \in N} \left[ c_j \min\limits_{i \in BS(j)} \dfrac{\beta_{ij}}{c_i} \right]
\leq
\sum_{(i,j) \in E} \beta_{ij} \, w_{ij}  
\leq 
\rho \sum_{j \in N} \left[ c_j \max\limits_{i \in BS(j)} \dfrac{\beta_{ij}}{c_i} \right],
$$

\item $$
\left| \left\{ i \in N: \ \exists\ (i,j) \in E \ \text{ with } c_i > \rho c_j \right\} \right|
\leq 
F(w)
\leq 
n,
$$
where $F(w) := | \{ i \in N: \ \exists\ (i,j) \in E \ \text{ with } w_{ij} \neq 1 \}|$ is the objective function of problem (P5).
\end{enumerate}	
\end{theorem}

\begin{proof} \
\begin{enumerate}[a)]

\item Since Lemma~\ref{l:bound} a) implies that $w_{ij}-1 \leq \rho c_j/c_i - 1$ holds for any $(i,j) \in E$, we get
\begin{align}\label{e:|w-1|}
|w_{ij}-1| = 
\max\{ 1-w_{ij} , w_{ij}-1 \}
\leq
\max\left\{ 1 , \dfrac{\rho c_j}{c_i} - 1 \right\},
\quad \forall\ (i,j) \in E.
\end{align}
Therefore, we have
$$
\| w - \mathbf{1} \|_1
=
\sum_{(i,j) \in E} |w_{ij}-1|
\leq 
\sum_{(i,j) \in E} \max\left\{ 1 , \ \dfrac{\rho c_j}{c_i} - 1 \right\}.
$$
Moreover, given any vector $x_0$, the minimum distance with respect to the $p$-norm between $x_0$ and the points belonging to the hyperplane $a^\top x=b$ is given by the following formula (see, e.g., \cite{Man99}):
$$
	\min \left\{ 
	\| x-x_0 \|_p : \ a^\top x = b
	\right\}
	=
	\dfrac{|a^\top x_0 - b|}{\|a\|_q},
$$
where $1/p+1/q = 1$. 
For any $j \in N$, we denote the subvector $w^j = (w_{ij})_{i \in BS(j)}$ and we remark that in the linear system~\eqref{e:Bw=rhoc} the variables of $w^j$ are present only in the equation associated to $j$. 
Therefore, we have the following chain of equalities and inequalities:
\begin{align*}
	& \min \left\{ \|w-\mathbf{1}\|_1 : \ w \in W\right\}
	\geq \min \left\{ \|w-\mathbf{1}\|_1 : \ A^{\top}_w c = \rho c \right\}
	\\
	& \qquad = \min \left\{ \sum_{j \in N} \sum_{i \in BS(j)} |w_{ij}-1| : \ \sum_{i \in BS(j)} c_i w_{ij} = \rho \, c_j,
	\quad \forall\ j \in N \right\}
	\\
	& \qquad = \sum_{j \in N} 
	\min \left\{ \|w^j-\mathbf{1}\|_1 : \ \sum_{i \in BS(j)} c_i w_{ij} = \rho \, c_j \right\}
	\\
	& \qquad = \sum_{j \in N} 
	\dfrac{\left| \rho \, c_j - \sum_{i \in BS(j)} c_i \right|}{\max_{i \in BS(j)} c_i},
\end{align*}  
where the second equality is due to the decomposability of the optimization problem.

\item Lemma~\ref{l:bound} b) implies that 
\begin{align*}
\| w - \mathbf{1} \|^2_2 & = \sum_{(i,j) \in E} (w_{ij}-1)^2
= \sum_{(i,j) \in E} w_{ij}^2 - 2 \sum_{(i,j) \in E} w_{ij} + m
\\
& = \sum_{j \in N} \sum_{i \in BS(j)} w_{ij}^2 - 2 \sum_{j \in N} \sum_{i \in BS(j)} w_{ij} + m
\\
& \leq \sum_{j \in N} \left( \sum_{i \in BS(j)} w_{ij} \right)^2 - 2 \sum_{j \in N} \sum_{i \in BS(j)} w_{ij} + m
\\
& \leq \sum_{j \in N} \left[ \dfrac{\rho \,c_j}{\min\limits_{i \in BS(j)} c_i} \right]^2 
- 2 \sum_{j \in N} \dfrac{\rho \,c_j}{\max\limits_{i \in BS(j)} c_i} + m
\end{align*}
Arguing as in a), we get
\begin{align*}
	& \min \left\{ \|w-\mathbf{1}\|^2_2 : \ w \in W\right\}
	\geq \min \left\{ \|w-\mathbf{1}\|^2_2 : \ A^{\top}_w c = \rho c \right\}
	\\
	& \qquad = \min \left\{ \sum_{j \in N} \sum_{i \in BS(j)} (w_{ij}-1)^2 : \ \sum_{i \in BS(j)} c_i w_{ij} = \rho \, c_j,
	\quad \forall\ j \in N \right\}
	\\
	& \qquad= \sum_{j \in N} 
	\min \left\{ \|w^j-\mathbf{1}\|^2_2 : \ \sum_{i \in BS(j)} c_i w_{ij} = \rho \, c_j \right\}
	\\
	& \qquad= \sum_{j \in N} 
	\left[ 
	\dfrac{\left| \rho \, c_j - \sum_{i \in BS(j)} c_i \right|}{ \| (c_i)_{i \in BS(j)} \|_2}
	\right]^2
	\\
	& \qquad= \sum_{j \in N} 
	\dfrac{\left( \rho \, c_j - \sum_{i \in BS(j)} c_i \right)^2}{ \sum_{i \in BS(j)} c_i^2}.
\end{align*}

\item It follows from~\eqref{e:|w-1|} that
$$
\| w - \mathbf{1} \|_{\infty} 
=
\max_{(i,j) \in E} |w_{ij}-1|
\leq 
\max \left\{ 1 , \ \max_{(i,j) \in E} \left[ \dfrac{\rho c_j}{c_i} - 1 \right]\right\}.
$$
Arguing as in a), we get
\begin{align*}
	& \min \left\{ \|w-\mathbf{1}\|_\infty : \ w \in W\right\}
	\geq \min \left\{ \|w-\mathbf{1}\|_\infty : \ A^{\top}_w c = \rho c \right\}
	\\
	& \qquad = \min \left\{ \max_{j \in N} \max_{i \in BS(j)} |w_{ij}-1| : \ \sum_{i \in BS(j)} c_i w_{ij} = \rho \, c_j,
	\quad \forall\ j \in N \right\}
	\\
	& \qquad = \max_{j \in N} \left[
	\min \left\{ \|w^j-\mathbf{1}\|_\infty : \ \sum_{i \in BS(j)} c_i w_{ij} = \rho \, c_j \right\}
	\right]
	\\
	& \qquad = \max_{j \in N} 
	\dfrac{\left| \rho \, c_j - \sum_{i \in BS(j)} c_i \right|}{\sum_{i \in BS(j)} c_i}.
\end{align*}  

\item It follows from~\eqref{e:Bw=rhoc} that
\begin{align*}
\sum_{(i,j) \in E} \beta_{ij} w_{ij} 
& = \sum_{j \in N} \sum_{i \in BS(j)} \beta_{ij} w_{ij}
= \sum_{j \in N} \sum_{i \in BS(j)} \dfrac{\beta_{ij}}{c_i} c_i w_{ij}
\\
& \leq \sum_{j \in N}
\left[ \max_{i \in BS(j)} \dfrac{\beta_{ij}}{c_i} \right]  \sum_{i \in BS(j)}  c_i w_{ij}
= 
\rho \sum_{j \in N}
\left[ \max_{i \in BS(j)} \dfrac{\beta_{ij}}{c_i} \right]   c_j,
\end{align*}
and
\begin{align*}
	& \sum_{(i,j) \in E} \beta_{ij} w_{ij} 
	= 
	\sum_{j \in N} \sum_{i \in BS(j)} \dfrac{\beta_{ij}}{c_i} c_i w_{ij}
	\\
	& \qquad
	\geq 
	\sum_{j \in N}
	\left[ \min_{i \in BS(j)} \dfrac{\beta_{ij}}{c_i} \right]  \sum_{i \in BS(j)}  c_i w_{ij}
	= 
	\rho \sum_{j \in N}
	\left[ \min_{i \in BS(j)} \dfrac{\beta_{ij}}{c_i} \right]   c_j,
\end{align*}

\item The upper bound is obvious. 
For any arc $(i,j) \in E$ such that $c_i > \rho c_j$, Lemma~\ref{l:bound} guarantees that $w_{ij} \leq \rho c_j / c_i < 1$ holds for any $w \in W$. 
Therefore, we get the inclusion
$$
	\left\{ i \in N: \ \exists\ (i,j) \in E \ \text{ with } c_i > \rho c_j \right\} 
	\subseteq
	\left\{ i \in N: \ \exists\ (i,j) \in E \ \text{ with } w_{ij} \neq 1 \right\},
$$
that implies the thesis. 
\end{enumerate}	
\end{proof}
In the following result we prove that problem (P6) can be solved in closed form. 

\begin{theorem}
Let $0 < \varepsilon < \min\left\{ 1 , \ \min\limits_{j \in N} \dfrac{\rho c_j}{\sum_{i \in BS(j)} c_i} \right\}$.
For any $j \in N$, we set $k_j:=|BS(j)|$, order the set $\{c_i: \ i \in BS(j)\}$ in non-decreasing order as $c_{i_1} \leq c_{i_2} \leq  \dots \leq c_{i_{k_j}}$ and denote
$$
h_j := 
\max\left\{
h \in \{0,\dots,k_j-1\}: \ \sum_{r=1}^h c_{i_r}  + \sum_{r=h+1}^{k_j} \varepsilon c_{i_r} \leq \rho c_j
\right\}.
$$
An optimal solution $w^*$ of problem (P6) is given by the following rule: for any $j \in N$, if $\rho c_j = \sum_{i \in BS(j)} c_i$, then $w^*_{ij}=1$ for any $i \in BS(j)$; else 
$$
	w^*_{i_r, j} =
	\begin{cases}
	1 & \text{for any } r = 1,\dots,h_j,
	\\
	\dfrac{\rho c_j - \sum_{s=1}^{h_j} c_{i_s} 
	- \sum_{s=h_j+2}^{k_j} \varepsilon c_{i_s}}{c_{i_r}} & \text{if } r = h_j+1,
	\\
	\varepsilon & \text{for any } r = h_j+2,\dots,k_j,
	\end{cases}
$$
and the corresponding optimal value is equal to $m - \sum_{j \in N} d_j$, where
$$
d_j = 
\begin{cases}
k_j & \text{if } \rho c_j = \sum_{i \in BS(j)} c_i,
\\
h_j & \text{otherwise}.
\end{cases}	
$$
\end{theorem}

\begin{proof}
First, we prove that $w^* \in W$. 
In fact, $B w^* = \rho c$ holds by definition of $w^*$ and $w^* \geq \varepsilon \mathbf{1}$ follows from the definition of $h_j$. 
The number of arcs $(i,j)$ with weights $w^*_{ij}$ different from 1 is equal to
\begin{align*}
	m - |\{ (i,j) \in E: \ w^*_{ij}=1 \}|
	= 
	m - \sum_{j \in N} | \{ i \in BS(j): \ w^*_{ij}=1 \} |
	=
	m - \sum_{j \in N} d_j.
\end{align*}
Let $w \in W$ be any feasible solution. 
Notice that  problem (P6) is decomposable into $n$ problems, one for each node $j \in N$. 
Given $j \in N$, we denote $I=\{ i \in BS(j):\ w_{ij}=1 \}$ and we prove that 
$|I| \leq d_j$ for any $j \in N$.
Suppose by contradiction that $|I| \geq d_j+1$. Then, one has $d_j=h_j$. Moreover, we have
\begin{align*}
	| \{ i_r \in I:\ r \geq h_j+2 \} | 
	& = |I| - | \{ i_r \in I: \ r \leq h_j+1\} |
	\\
	& \geq h_j + 1 - | \{ i_r \in I: \ r \leq h_j+1\} |
	\\
	& = | \{ i_r \notin I: \ r \leq h_j+1\} |.
\end{align*}
Since $c_{i_r}$ are in non-decreasing order, we get 
$$
	\sum_{\substack{i_r \in I: \\ r \geq h_j+2}} c_{i_r}
	\geq 
	\sum_{\substack{i_r \notin I: \\ r \leq h_j+1}} c_{i_r}.
$$
Hence, we have
$$
(1-\varepsilon) \sum_{\substack{i_r \in I: \\ r \geq h_j+2}} c_{i_r}
\geq 
(1-\varepsilon) \sum_{\substack{i_r \notin I: \\ r \leq h_j+1}} c_{i_r},
$$
that is equivalent to
\begin{align}\label{e:P6proof}
\sum_{\substack{i_r \in I: \\ r \geq h_j+2}} c_{i_r}
+
\varepsilon \sum_{\substack{i_r \notin I: \\ r \leq h_j+1}} c_{i_r}
\geq 
\sum_{\substack{i_r \notin I: \\ r \leq h_j+1}} c_{i_r}
+
\varepsilon \sum_{\substack{i_r \in I: \\ r \geq h_j+2}} c_{i_r}.
\end{align}
The definition of $h_j$ and~\eqref{e:P6proof} imply the following chain of equalities and inequalities:
\begin{align*}
\rho c_j & = \sum_{i \in BS(j)} c_i w_{ij}
\\
& = \sum_{i \in I} c_i w_{ij} + \sum_{i \in BS(j) \setminus I} c_i w_{ij}
\\
& \geq \sum_{i \in I} c_i + \sum_{i \in BS(j) \setminus I} \varepsilon c_i
\\
& = \sum_{\substack{i_r \in I \\ r \leq h_j+1}} c_{i_r}
+
\sum_{\substack{i_r \in I \\ r \geq h_j+2}} c_{i_r}
+
\sum_{\substack{i_r \notin I \\ r \leq h_j+1}} \varepsilon c_{i_r}
+
\sum_{\substack{i_r \notin I \\ r \geq h_j+2}} \varepsilon c_{i_r}
\\
& \geq \sum_{\substack{i_r \in I \\ r \leq h_j+1}} c_{i_r}
+
\sum_{\substack{i_r \notin I: \\ r \leq h_j+1}} c_{i_r}
+
\sum_{\substack{i_r \in I: \\ r \geq h_j+2}} \varepsilon c_{i_r}
+
\sum_{\substack{i_r \notin I \\ r \geq h_j+2}} \varepsilon c_{i_r}
\\
& = \sum_{\substack{r \leq h_j+1}} c_{i_r}
+
\sum_{\substack{r \geq h_j+2}} \varepsilon c_{i_r}
\\
& > \rho c_j,
\end{align*}
which is impossible.
Therefore, we have
\begin{align*}
	|\{ (i,j) \in E: \ w_{ij} \neq 1 \}|
	& = m - \sum_{j \in N} | \{ i \in BS(j): \ w_{ij}=1 \} |
	\\
	& \geq
	 m - \sum_{j \in N} d_j
	 = 
	 |\{ (i,j) \in E: \ w^*_{ij} \neq 1 \}|,
\end{align*}
thus $w^*$ is an optimal solution of (P6).
\end{proof}


\section{Numerical experiments}
\label{numexp}

We now show some numerical experiments on the optimization problems (P1)--(P6) described in Section~\ref{sec:optim} for some sample graphs. 
All the problems have been solved by means of the MATLAB optimization toolbox.
 
We start with a small directed graph with 8 nodes and 20 arcs (see Fig.~\ref{fig:2}), whose adjacency matrix is
$$
A = \begin{pmatrix}
0 &    1 &    1 &    0 &    1 &    0 &    0 &    1
\\
0 &    0  &   0  &   0  &   1  &   1  &   0  &   0
\\
1 &    1   &  0   &  0   &  0   &  0   &  0   &  0
\\
0  &   0    & 1    & 0    & 0    & 0    & 1    & 1
\\
0   &  0    & 1    & 1    & 0    & 0    & 1    & 1
\\
0    & 0    & 0    & 0    & 1    & 0    & 1    & 0
\\
0    & 0    & 0    & 0    & 1    & 0    & 0    & 0
\\
1    & 1    & 0    & 0    & 0    & 0    & 0    & 0
\end{pmatrix}.
$$

\begin{figure}[tbp]
\begin{center}
\includegraphics[width=0.99\textwidth]{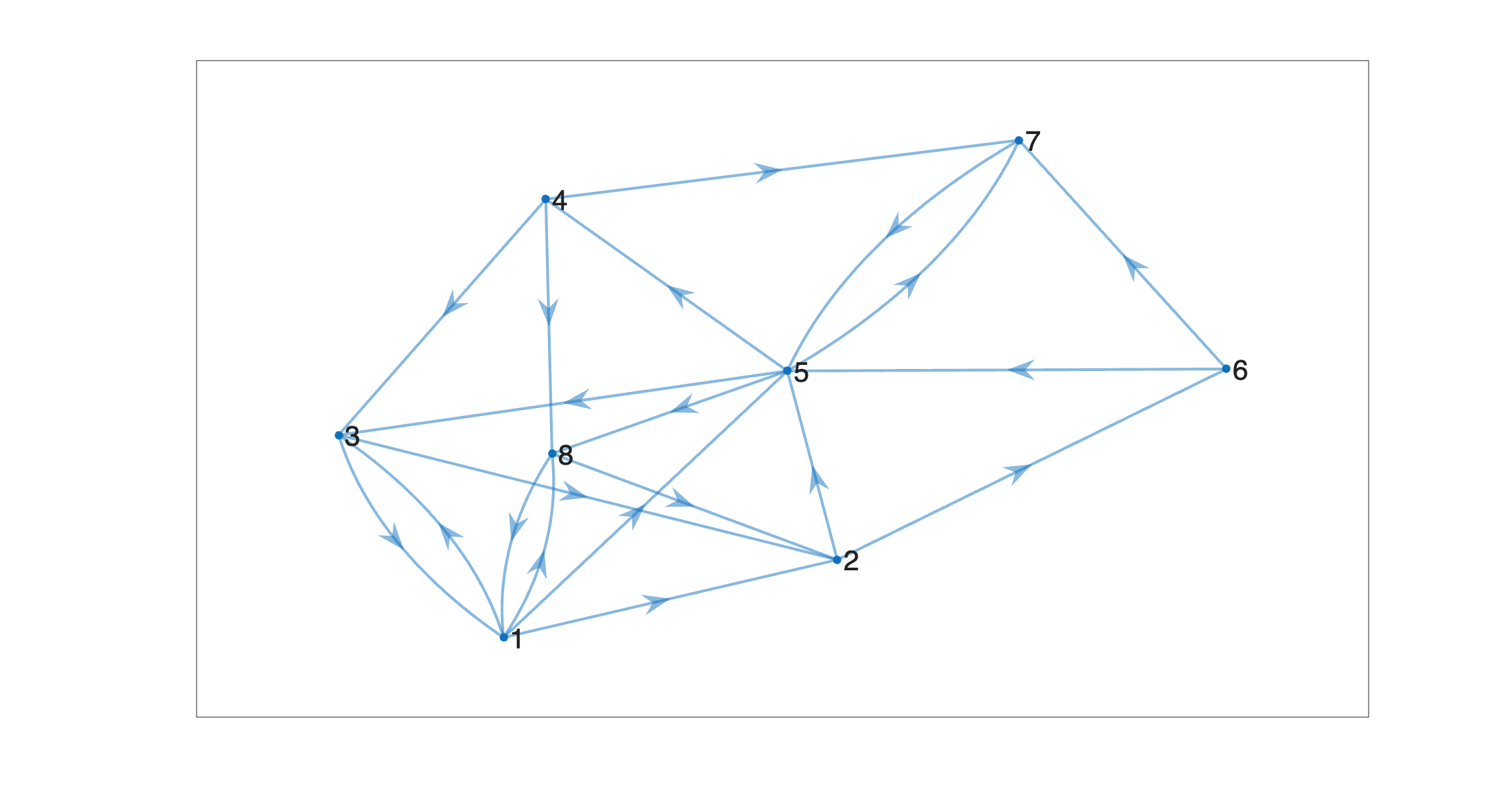}
\end{center}
\caption{Directed graph with 8 nodes and 20 arcs.}
\label{fig:2}
\end{figure}

We notice that the graph is strongly connected, the spectral radius $\rho_0 = \rho(A^\top)=2.5367$ and the centrality eigenvector is
$$
c_0 = ( 0.1138,  \ 0.1586,    \ 0.1443,    \ 0.0713,  \   0.1810,  \   0.0625, \    0.1241,   \ 0.1443).
$$
We now consider the inverse eigenvector centrality problem with $\rho=\rho_0$ and 
$$
c = ( 0.1586,  \ 0.1138,    \ 0.1443,    \ 0.0713,  \   0.1810,  \   0.0625, \    0.1241,   \ 0.1443).
$$
We note that $c$ coincides with $c_0$ except for the first two components, which are the first two components of $c_0$ in reverse order. 
We set $\beta_{ij}=1$ for any $(i,j) \in E$ in problem (P4) and $\varepsilon=10^{-3}$ in all the problems (P1)--(P6). 
The optimal solutions of (P1)--(P6) are reported in Table~\ref{tab:1} and shown in Fig.~\ref{fig:3}, where the arcs with a weight different from 1 are highlighted in red. 
We note that each solution differs from all the others, except for solutions of (P1) and (P6), which are very similar.
As expected, the solution of (P1) is more sparse (i.e., with few weights other than 1) than the solutions of (P2) and (P3).
The weights of the solution of (P4) are in a wider range than the other solutions: several values are equal to $\varepsilon$, while other values are greater than 2. 

\begin{table}[tbp]
\begin{center}
\begin{tabular}{c@{\quad\qquad}c@{\qquad}c@{\qquad}c@{\qquad}c@{\qquad}c@{\qquad}c}
\toprule
 & Opt. sol. & Opt. sol. & Opt. sol. & Opt. sol. & Opt. sol. & Opt. sol. \\
 Arc   & (P1) & (P2) & (P3) & (P4) & (P5) & (P6) \\
\midrule
(1,2) &    0.0010 &   0.6234 &   0.6058 &   1.8176  &  0.9088 &   0.0010 \\
(1,3) &    1.0000 &   0.8871 &   1.3442 &   0.0010 &   0.7173 &   1.0000 \\
(1,5) &    1.0000 &   1.0000 &   0.6058 &   2.8920 &   0.0010 &   1.0000 \\
(1,8) &   1.0000 &   0.8871  &  1.3442 &   0.0010  &  0.7173  &  1.0000 \\
(2,5) &   1.0000 &   1.0000 &   1.3363 &   0.0010  &  2.3928  &  1.0000 \\
(2,6) &   1.3942 &   1.3942 &   1.3942 &   1.3942 &   1.3942  &  1.3942 \\
(3,1) &  1.7884  &  1.3942  &  1.3942  &  0.0010  &  1.7884  &  1.7884 \\
(3,2) &   1.0000 &   0.6574 &   0.7283 &   0.0010 &   0.0010 &   0.9989 \\
(4,3) &   1.0000 &   0.9492 &   0.6058 &   0.0010 &   1.0000 &   1.0000 \\
(4,7) &   1.0000 &   1.0000 &   0.6058 &   0.0010  &  1.0000 &   1.0000 \\
(4,8) &   1.0000 &   0.9492 &   0.6058 &   0.0010  &  1.0000 &   1.0000 \\
(5,3) &   0.7521 &   0.8712 &   0.6058 &   2.0217 &   1.0000 &   0.7521 \\
(5,4) &  1.0000  &  1.0000  &  1.0000  &  1.0000  &  1.0000  &  1.0000 \\
(5,7) &   1.0000 &   1.0000 &   1.2916 &   1.7390 &   1.0000 &   1.0000 \\
(5,8) &   0.7521 &   0.8712 &   0.6058 &   2.0217 &   1.0000 &   0.7521 \\
(6,5) &   1.0000 &   1.0000 &   0.6058 &   0.0010 &   1.0000 &   1.0000 \\
(6,7) &  1.0000  &  1.0000  &  0.6058  &  0.0010  &  1.0000  &  1.0000 \\
(7,5) &   1.0000 &   1.0000 &   1.3942 &   0.0010 &   1.0000 &   1.0000 \\
(8,1) &   1.0000 &   1.3942 &   1.3942 &   2.7874 &   1.0000 &   1.0000 \\
(8,2) &   0.9989 &   0.6574 &   0.6058  &  0.0010 &   1.0000 &   1.0000 \\
\bottomrule
\end{tabular}
\end{center}
\caption{Optimal solutions of problems (P1)--(P6) related to the graph in Fig.~\ref{fig:2}.}
\label{tab:1}
\end{table}

\begin{figure}[tbp]
	\begin{center}
		\includegraphics[width=0.99\textwidth]{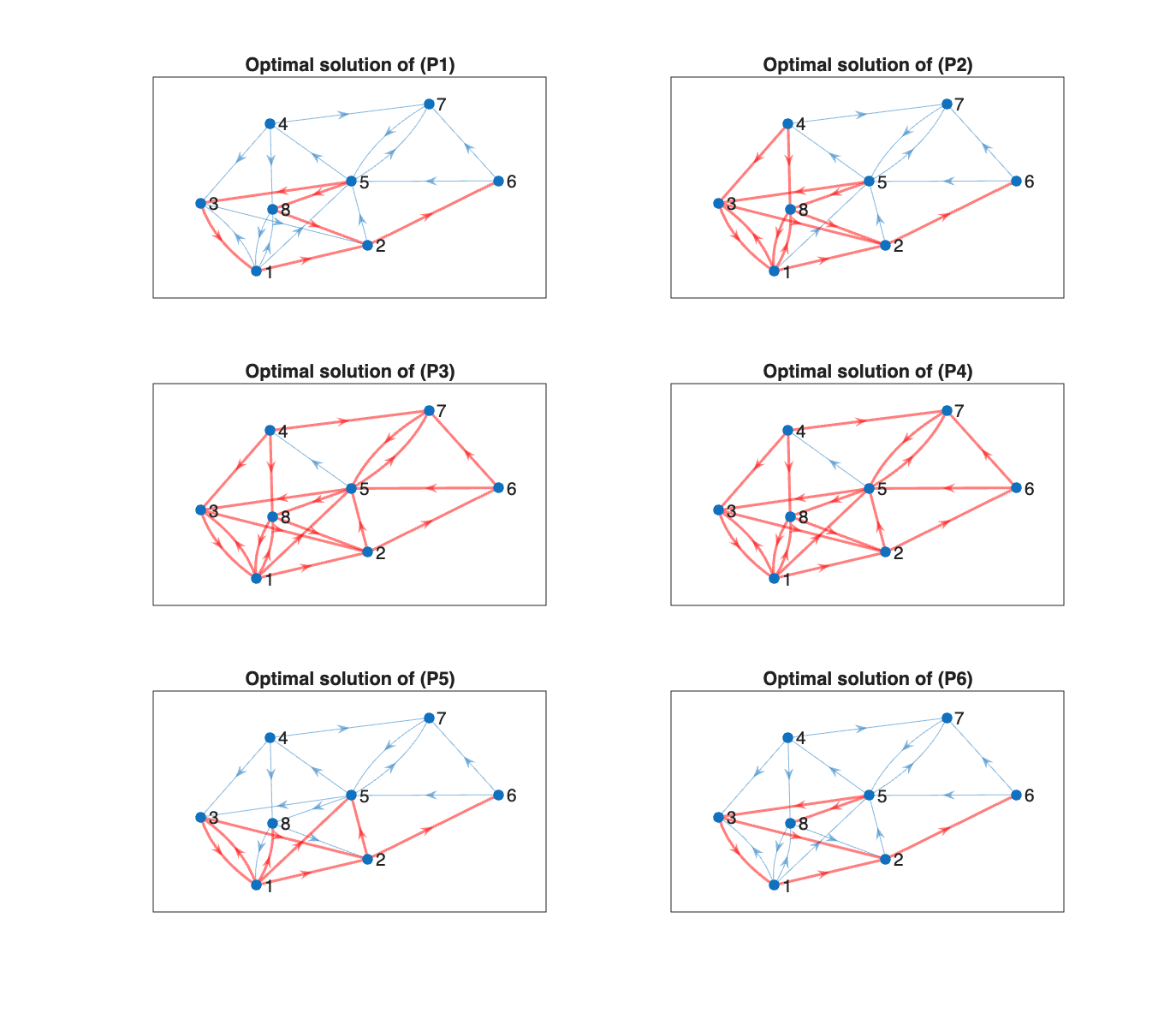}
	\end{center}
	\caption{Comparison between the optimal solutions of problems (P1)-(P6). The arcs with a weight different from 1 are highlighted in red.}
	\label{fig:3}
\end{figure}

Next, we solved problems (P1)--(P6) on three medium-sized networks which are well known in the literature (see, e.g.,~\cite{Pe20}) and described in Table~\ref{tab:instances}.
As in the previous example, we set $\rho=\rho_0$ and $c$ equal to $c_0$ except for the first two components, which are the first two components of $c_0$ in reverse order. 

\begin{table}[bp]
\begin{center}
\begin{tabular}{l@{\qquad}c@{\qquad}c}
\toprule
Network & Nodes & Links \\ 
\midrule
Rhesus monkey~\cite{Sa72} & 16 & 111 
\\
High tech company~\cite{Kra87} & 21 & 232 
\\
Bison~\cite{Lot79} & 26 & 314 
\\
\bottomrule
\end{tabular}	
\end{center}
\caption{Test networks.}
\label{tab:instances}
\end{table}

To briefly highlight the difference between the optimal solutions of (P1)--(P6), we show in Tables~\ref{tab:monkey}, \ref{tab:htc} and \ref{tab:bison} the value of each objective function calculated at the optimal solutions of the six problems.
The results show that the criteria used in the six optimization problems are quite different from each other.
Indeed, the values of each objective function calculated in the optimal solutions of the six problems vary within a fairly wide range.

\begin{table}[tbp]
\begin{center}
\begin{tabular}{l@{\qquad}r@{\quad}r@{\quad}r@{\quad}r@{\quad}r@{\qquad}r}
	\toprule
	                & Opt. sol. & Opt. sol. & Opt. sol. & Opt. sol. & Opt. sol. & Opt. sol. \\
	                &      (P1) &      (P2) &      (P3) &      (P4) &      (P5) &      (P6) \\ \midrule
	Obj. func. (P1) &    7.4838 &    9.5436 &   61.7243 &  155.2115 &   14.8129 &   18.4516 \\
	Obj. func. (P2) &    3.1312 &    1.8597 &    6.0634 &   19.3768 &    4.3120 &    9.1841 \\
	Obj. func. (P3) &    2.5593 &    0.7894 &    0.6091 &    7.6046 &    2.5593 &    8.5505 \\
	Obj. func. (P4) &  112.9506 &  113.8430 &  104.7592 &   76.4015 &  113.2091 &  123.9184 \\
	Obj. func. (P5) &    3.0000 &   15.0000 &   16.0000 &   16.0000 &    3.0000 &    8.0000 \\
	Obj. func. (P6) &   11.0000 &   69.0000 &  111.0000 &  111.0000 &   16.0000 &   11.0000 \\ \bottomrule
\end{tabular}
\end{center}
\caption{Rhesus monkey network: values of the objective functions of problems (P1)-(P6) calculated at the optimal solutions of the same problems.}
\label{tab:monkey}
\end{table}

\begin{table}[tbp]
	\begin{center}
		\begin{tabular}{l@{\qquad}r@{\quad}r@{\quad}r@{\quad}r@{\quad}r@{\qquad}r}
			\toprule
			& Opt. sol. & Opt. sol. & Opt. sol. & Opt. sol. & Opt. sol. & Opt. sol. \\
			&      (P1) &      (P2) &      (P3) &      (P4) &      (P5) &      (P6) \\ 
			\midrule
	Obj. func. (P1) &    7.6652 &  11.3575 &  78.2698 & 330.1895 &  29.8715 &  16.2246 \\
	Obj. func. (P2) & 3.6263 &   1.7079 &   5.2441 &  31.3891  &  6.1077 &  11.0225 \\
	Obj. func. (P3) & 3.1266 &   0.5505 &   0.3564  & 10.8915 &   3.1266 &  10.8300 \\
	Obj. func. (P4) & 231.6782 & 231.8134 & 226.6423 & 140.6115 & 231.4700 & 240.2376 \\
	Obj. func. (P5) & 4.0000 &  21.0000  & 21.0000 &  21.0000 &   4.0000 &   6.0000 \\
	Obj. func. (P6) & 11.0000 & 103.0000 & 232.0000 & 232.0000 &  37.0000 &  11.0000 \\
			\bottomrule
		\end{tabular}
	\end{center}
	\caption{High tech company network: values of the objective functions of problems (P1)-(P6) calculated at the optimal solutions of the same problems.}
	\label{tab:htc}
\end{table}

\begin{table}[tbp]
	\begin{center}
		\begin{tabular}{l@{\qquad}r@{\quad}r@{\quad}r@{\quad}r@{\quad}r@{\qquad}r}
			\toprule
			& Opt. sol. & Opt. sol. & Opt. sol. & Opt. sol. & Opt. sol. & Opt. sol. \\
			&      (P1) &      (P2) &      (P3) &      (P4) &      (P5) &      (P6) \\ 
			\midrule
	Obj. func. (P1) &    9.7070 &  12.8607 & 276.5547 & 464.8477 &  80.8821 &  39.0658\\
	Obj. func. (P2) & 4.7889 &   2.8468 &  15.7666 &  42.4710 &  11.8447 &  22.9597 \\
	Obj. func. (P3) & 4.2933 &   1.3231  &  0.9119 &  12.0516 &   4.6204  & 16.6946 \\
	Obj. func. (P4) & 314.2816 & 315.5194 & 281.5060 & 203.4237 & 313.2941 & 343.6404 \\
	Obj. func. (P5) & 5.0000 &  24.0000 &  26.0000 &  26.0000 &   5.0000 &   8.0000 \\
	Obj. func. (P6) & 9.0000 &  54.0000 & 313.0000 & 313.0000 &  68.0000  &  9.0000 \\
			\bottomrule
		\end{tabular}
	\end{center}
	\caption{Bison network: values of the objective functions of problems (P1)-(P6) calculated at the optimal solutions of the same problems.}
	\label{tab:bison}
\end{table}


\section{Conclusions}
We investigated six optimization-based formulations of the inverse eigenvector centrality problem, addressing the inherent non-uniqueness of arc-weight realizations associated with a prescribed centrality profile. These formulations provide a unified framework that can be interpreted either as centralized design strategies or as endogenous mechanisms shaping weighted network structures. Our analytical results establish basic properties of the solutions, while numerical experiments illustrate how different objectives lead to distinct network realizations despite inducing the same centrality vector.

Several directions for future research naturally arise, such as extensions to undirected graphs, time-varying networks, and  alternative centrality measures. From a modeling perspective, incorporating additional structural constraints -- such as sparsity, capacity limits, or fairness considerations -- could further enhance applicability. Finally, these formulations naturally lend themselves to equilibrium-based extensions, in which edge weights emerge from strategic interactions rather than centralized optimization (see, e.g., \cite{PasRac21b,PasRac24}).


\subsection*{Acknowledgement}

The authors are members of the Gruppo Nazionale per l'Analisi Matema\-tica, la Probabilità e le loro Applicazioni (GNAMPA - National Group for Mathematical Analysis, Probability and their Applications) of the Istituto Nazionale di Alta Matematica (INdAM - National Institute of Higher Mathematics).
\\
The research was supported by the MUR research
programs PRIN2022 founded by the European Union - Next Generation EU (Project ``ACHILLES, eco-sustAinable effiCient tecHdrIven Last miLE logiStics'', CUP: E53D23005640006). 
It was also partially supported by the research project ``Programma ricerca di Ateneo UNICT 2024-26 NOVA - Network Optimization and Vulnerability Assessment'' of the University of Catania.



\begin{thebibliography}{99}

\bibitem{BerPle94}  
Bermann, A., Plemmons, R.J.:
Nonegative matrices in mathematical sciences.
Academic Press, New York (1979)

\bibitem{Bon72}  
Bonacich, P.: 
Factoring and weighting approaches to status scores and clique identification.
Journal of Mathematical Sociology, 2(1), 113–120 (1972)


\bibitem{Bon87} 
Bonacich, P.: 
Power and centrality: a family of measures. 
Am. J. Sociol. 92, 1170-1182 (1987)


\bibitem{BriPag98} 
Brin, S., Page, L.: 
The Anatomy of a Large-Scale Hypertextual Web Search Engine, Computer Networks and ISDN Systems, vol. 30, no. 1-7, pp. 107–117, 1998.

\bibitem{Cipolla25} 
Cipolla, S., Durastante, F., Meini, B.:
Enforcing Katz and PageRank Centrality Measures in Complex Networks.
SIAM Journal on Mathematics of Data Science 7, 1514-1539 (2025)

\bibitem{Fre79}  
Freeman, L.C.: 
Centrality in Social Networks: Conceptual Clarification.
Social Networks 1 (3), 215-239 (1979)

\bibitem{Horn} 
Horn, R.A., Johnson, C.R.: 
Matrix Analysis. 
Cambridge University Press, Cambridge (2013)

\bibitem{Kra87} 
Krackhardt, D.:
Cognitive social structures. 
Social Networks 9, 104-134 (1987)

\bibitem{Lan95} 
Landau, E.:  
Zur Theorie der Turniere. 
Zeitschrift f\"{u}r Mathematik und Physik 46, 447-458 (1895)

\bibitem{Lot79} 
Lott, D.F.:
Dominance relations and breeding rate in mature male American bison. 
Zeitschrift f\"{u}r Tierpsychologie 49(4), 418-432 (1979)

\bibitem{Man99}  
Mangasarian, O.: 
Arbitrary-norm separating plane.
Operations Research Letters 24, 15-23 (1999) 

\bibitem{Min88}  
Minc, H.: 
Nonnegative Matrices.
Wiley-Interscience Series in Discrete Mathematics and Optimization.
John Wiley \& Sons, New York, 1988

\bibitem{New10} 
Newman, M.E.J.:
Networks: An Introduction. 
Oxford University Press, Oxford, UK, 2010.

\bibitem{Lat_etal12} 
Nicosia, V., Criado, R., Romance, M., Russo, G., Latora, V.:
Controlling centrality in complex networks. 
Scientific Reports, vol. 2, Article 218, 2012.

\bibitem{PasRac21b} 
Passacantando M., Raciti, F.: 
A note on generalized Nash games played on networks. 
In: T.M. Rassias (ed), Nonlinear Analysis, Differential Equations, and Applications, Springer Optimization and Its Applications, vol. 173, 365-380 (2021).

\bibitem{PasRac24} 
Passacantando, M., Raciti, F.: 
A continuity result for the Nash equilibrium of a class of network games. 
J. Nonlinear Var. Anal. 8, 167-179 (2024).

\bibitem{PasRac26} 
Passacantando, M., Raciti, F.:
Exploring network centrality through the lens of Game Theory. 
In: P. Pardalos and T. Rassias (eds), Convex and Variational Analysis with Applications: In Honor of Terry Rockafellar's 90th Birthday, Springer, to appear.

\bibitem{Pe20} 
Peixoto, T.P.:
The Netzschleuder network catalogue and repository. 
\url{https://networks.skewed.de}, (2020). 

\bibitem{Sa72} 
Sade, D.S.: 
Sociometrics of macaca mulatta I. linkages and cliques in grooming matrices. 
Folia Primatologica 18 (3-4), 196-223 (1972)

\end{thebibliography}
\end{document}